\documentclass[10pt,english,twocolumn]{article}

\usepackage{times}
\usepackage{graphicx} 
\usepackage{subfigure}

\usepackage[top=1in, bottom=1in, left=0.75in, right=0.75in]{geometry}

\usepackage[round]{natbib}

\usepackage{algorithm}
\usepackage{algorithmic}

\usepackage{hyperref}

\usepackage{amsmath}

\usepackage[final=true]{daranki}

\usepackage{paralist}

\usepackage{wrapfig}

\usepackage{epstopdf}

\usepackage{multirow}

\usepackage[]{mcode}

\usepackage{array}

\usepackage[capitalize,noabbrev]{cleveref}

\usepackage{tikz}

\usepackage{lipsum}

\usepackage{authblk}

\newcolumntype{L}[1]{>{\raggedright\let\newline\\\arraybackslash\hspace{0pt}}m{#1}}

\newcommand{\myverb}{\lstinline[breaklines,keepspaces,columns=fullflexible,postbreak=]}

\comment{
\begin{lstlisting}
\end{lstlisting}
}

\title{Private Disclosure of Information in Health Tele-monitoring}

\author{Daniel Aranki}
\author{Ruzena Bajcsy}
\affil{Electrical Engineering and Computer Science Department,\\
University of California, Berkeley,\\
Berkeley, CA 94720 USA\\
\{daranki, bajcsy\}@eecs.berkeley.edu
}
\date{}

\declaretheoremstyle[
notefont=\normalfont, notebraces={}{},
bodyfont=\normalfont\itshape,
headformat=\NAME~\NUMBER \NOTE
]{nopar}
\declaretheorem[style=nopar,sibling=lemma,name={Lemma}]{lemma2}
\declaretheorem[style=nopar,sibling=definition,name={Definition}]{definition2}

\newcounter{desccount}
\newcommand{\descitem}[1]{%
  \item[#1] \refstepcounter{desccount}\label{desc:#1}
}
\newcommand{\descref}[1]{(\hyperref[desc:#1]{\textbf{#1}})}

\renewcommand{\pr}[1]{\textbf{Pr}\left\{ {#1} \right\}}
\renewcommand{\pr}[1]{\m{\mathbb{P}}{#1}}

\AppendGraphicsExtensions{.tif}

\begin{document} 

\maketitle

\begin{abstract} 

We present a novel framework, called Private Disclosure of Information (PDI), which is aimed to prevent an adversary from inferring certain sensitive information about subjects using the data that they disclosed during communication with an intended recipient. We show cases where it is possible to achieve perfect privacy regardless of the adversary's auxiliary knowledge while preserving full utility of the information to the intended recipient and provide sufficient conditions for such cases. We also demonstrate the applicability of PDI on a real-world data set that simulates a health tele-monitoring scenario.
\end{abstract} 

\section{Introduction}
\label{sec:introduction}

\defcitealias{nahdo1996guide}{NAHDO, 1996}
\defcitealias{california2014california}{OSHPD, 2014}

Data collection and sharing is growing to unprecedented volumes. Some of the reasons for this phenomenon include the decrease in storage cost, the rise of social networks, the ubiquity of smartphones and law regulations. For example, in many states in the US, medical institutions are obliged to make demographics data public about their patients~(\citetalias{nahdo1996guide}; \Citealp{sweeney2002k}; \citetalias{california2014california}).

\citet{warner1965randomized} argues that the lack of privacy guarantees can cause subjects to be reluctant to share their data with data collectors (such as doctors, government agencies, researchers, etc.) or even result in subjects providing false information. Therefore, subjects need to be assured that their privacy will be preserved throughout the whole process of data collection and use.

One of the emerging areas with growing interest to collect sensitive personal and private data is health tele-monitoring. In this setting, a technology is used to collect health-related data about patients, which are later submitted to a medical staff for monitoring. The data are then used to assess the health status of patients and provide them with feedback and/or intervention. Research indicates that such technologies can improve readmission rates and lower overall costs~\citep{clark2007telemonitoring, chaudhry2010telemonitoring, inglis2010structured, giamouzis2012telemonitoring, aranki2014continouos}. In such scenarios, the collected data are usually of sensitive nature from a privacy point of view and therefore privacy preserving technologies are needed in order to protect patients' privacy and increase compliance.

There are multiple stages in the life-cycle of data, including
\begin{inparaenum}[\itshape i\upshape)]
\item the disclosure (or submission) of the data by the subjects to the data collector;
\item the processing of the data;
\item the analysis; and/or
\item the publishing of (often a privatized version of) the data or some findings based on them.
\end{inparaenum}
In this paper we focus on the phase of disclosure of privacy-sensitive data by the data owners. Our framework for Private Disclosure of Information (PDI) is thus aimed to prevent an adversary from inferring certain sensitive information about the subject using the data that were disclosed during communication with an intended recipient. This is analogous to the problem of attribute linkage in statistical database privacy.

In traditional encryption approaches to maintaining privacy, it is often implicitly assumed that the data themselves \emph{are} the private information. However, in more general scenarios, the data \emph{can be used} to infer some private information about the subjects for which the data apply. For example, respiration rate by itself might not be considered private information. However, if the data from the collected respiration rate are used to infer whether the individual is a smoker or not, they become sensitive information. One can argue that because the information about whether someone smokes is private, the respiration rate data become private \emph{by implication}.

Under such circumstances, one should attempt to privatize the transmitted data in a way that reveals as little as possible about the private information to an adversary. In summary, our objective is to encode the transmitted data in order to hide another private piece of information. In the words of~\citet{sweeney2002k}: ``Computer security is not privacy protection." The converse is also true, privacy does not replace security. Our approach is therefore to be viewed as complementary to classical security approaches. For example, data can be privatized then encrypted.

The rest of this paper is organized as follows. In~\cref{sec:related_work} we provide a survey of the literature for related work. In~\cref{sec:problem_formulation} we provide the motivation to the problem and formulate it, followed by further analysis in~\cref{sec:further_analysis}. We then discuss implementation details of the learning problem in~\cref{sec:implementation} followed by experimental results in~\cref{sec:experiment}. Finally, we close by discussing our conclusions and future research directions in~\cref{sec:discussion}.

\section{Related Work}
\label{sec:related_work}

The study of privacy-preserving techniques and technologies in the fields of statistics, computer security and databases, and their intersections, dates back to at least~\citeyear{warner1965randomized} when~\citeauthor{warner1965randomized} proposed a randomization technique for conducting surveys and collecting responses for the purpose of statistical and population analysis. Since then, extensive privacy research in the fields above was conducted. Therefore, in the interest of brevity, we provide a brief overview of the areas of study related to our work and refer the reader to more comprehensive surveys in each area.

Recently, attention to privacy has been rising in the health-care domain with the spread of electronic health-records usage and the growing data sharing between medical institutions. It has been reported that consumers are expressing increasing concerns regarding their health privacy~\citep{bishop2005national,hsiao2012use}. Most of the research in privacy from the health community focuses on medical data publishing and is therefore database-centric. For a survey of results in this domain, we refer the reader to~\citep{gkoulalas2014publishing}.

In more general-purpose scenarios, the privacy of statistical databases and data publishing has been extensively studied. \citet{denning1983inference} presented some of the early threats related to inference in statistical databases and reviewed controls that are based on the lattice model~\citep{denning1976lattice}. \citet{duncan1989risk,duncan1986disclosure} studied methods for limiting disclosure and linkage risks in data publishing. \citet{sandhu1993lattice} provided a tutorial on lattice-based access controls for information flow security and privacy. Later, \citet{farkas2002inference} provided a survey of more results in the field of access controls to the inference problem in database security. For rigorous surveys in the fields of data publishing privacy and statistical databases privacy, we refer the reader to~\citep{adam1989security, fung2010privacy}.

Two semantic models of database privacy of growing interest in the privacy literature are $k$-anonymity~\citep{sweeney2002k} and differential privacy~\citep{dwork2006differential,dwork2008differential}. In $k$-anonymity, given a set of quasi-identifiers that can be used to re-identify subjects, a table is called $k$-anonymous if every combination of quasi-identifiers in the table appears in at least $k$ records. If a table is $k$-anonymous, assuming each individual has a single record in the table, then the probability of linking a record to an individual is at most $1/k$. Other extensions and refinements of $k$-anonymity have been proposed including $l$-diversity~\citep{machanavajjhala2007ldiversity}, $t$-closeness~\citep{li2007tcloseness} and others.

In differential privacy, the requirement is that the output of a statistical query should not be too sensitive to any single record in the database. Formally, given a statistical query $M$, then $M$ is $\epsilon$-differentially private if $\pr{\m{M}{D_1} \in S} \leq e^\epsilon \times \pr{\m{M}{D_2} \in S}$ for any two realizations $D_1$ and $D_2$ of the database such that $|D_1 \Delta D_2| = 1$ and all $S \subset Range(M)$, where $D_1 \Delta D_2$ is the symmetric difference between $D_1$ and $D_2$~\citep{dwork2006differential,dwork2008differential}. \citet{cormode2011personal} showed that sensitive attribute inference can be done on databases that are differentially private and $l$-diverse with similar accuracy.

As can be seen from the review above, most of the research in data-privacy is focused on privacy-preserving data publishing and privacy-preserving statistical databases. In contrast, in this work we focus on preventing adverserial statistical inference of a piece of private information based on the disclosed messages in an individual's information exchange scenario during communication.


\section{Problem Formulation}
\label{sec:problem_formulation}

\subsection{Notation}

We use the following shorthand notation for probability density (mass) functions. We always use a pair of a capital and a small symbols of the same letter for a random variable and a realization of it, respectively. For notation simplicity and conciseness, given random variables $X$ and $Y$, instead of writing $p_X(x)$ for the marginal density (mass) function of $X$ we simply write $p(x)$, and instead of writing $p_{X|Y}(x|y)$ for the conditional density (mass) function of $X$ given $Y$, we simply write $p(x|y)$.

\subsection{Motivation and Threat Model}

We are primarily motivated by the tele-monitoring setting. In this setting, a doctor wishes to monitor her patients remotely using a technology that can collect and transmit health-related data. The shared data are of sensitive nature because they can be used to infer private pieces of information like a health-condition or a disease. For example, updates about a patient's weight can lead to disclosure of obesity as it will be demonstrated in~\cref{sec:experiment}.

More generally, an information provider Bob wants to disclose a piece of information $x$ to some recipient Alice. Furthermore, the information $x$ can be used to infer some private information $c$ about Bob. However, there is no guarantee that the transmitted information will not be intercepted and potentially used for inference of the private information $c$ about Bob by an untrusted but passive eavesdropper Eve. Finally, in this setting, we assume that Alice is more certain about $c$ than Eve is. The problem at hand is delivering the information $x$ under these circumstances such that Alice can make full use of the information but that Eve's ability to infer $c$ about Bob, using the transmitted message, is minimized.

\begin{figure}[t]
\centerline{
\scalebox{.7}{
\begin{tikzpicture}[->,>=latex,shorten >=1pt,auto,node distance=2.5cm,
  thick,normal node/.style={circle,draw,font=\sffamily\Large\bfseries},observed node/.style={circle,fill=blue!20,draw,font=\sffamily\Large\bfseries}]
\node[normal node] (s) {$S$};
\node[normal node] (c) [right of=s] {$C$};
\node[normal node] (x) [below of=c] {$X$};
\node[normal node] (z) [right of=c] {$Z$};

\path[every node/.style={font=\sffamily\small}]
(s) edge node [left] {} (c)
	edge node [left] {} (x)
(c) edge node [left] {} (z)
	edge node [left] {} (x)
(x) edge node [left] {} (z);
\end{tikzpicture}
}
}
\caption{The Graphical Model of PDI} \label{fig:graphical-model}
\end{figure}
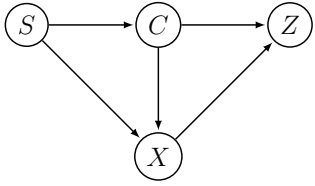

As a concrete example, consider the following scenario in health tele-monitoring. A patient Bob is trying to update his physician Alice about his weight and body mass index (BMI).\footnote{BMI is a measure of relative weight based on an individual's mass and height. Defined as $BMI \triangleq \frac{mass(kg)}{height(m)^2}$.} Since Alice is Bob's physician, she already knows the weight status category of Bob which he considers to be private information.\footnote{Weight status category indicates if an individual is underweight, overweight, obese or has a healthy weight.\comment{ This notion will be presented formally in~\autoref{sec:experiment}}} Eve, however, does not know Bob's weight status category \emph{a priori} but would like to learn it from the messages he sends to Alice. If Eve succeeds to listen in on the communication between Bob and Alice, Eve can, with some accuracy, infer the weight status category of Bob. Alice, being a considerate physician, wants to ensure the privacy of her patients. Alice decides to create an encoding scheme (that can be made public) for the communication such that the encoding is different per weight status group. Her objective is to make this encoding scheme ``as privacy-preserving as possible" in the sense of keeping her patients' weight status category information as private as possible to someone who does not know it a priori.

It is important to compare this scenario with the classical security approach. In classical security, the objective is to protect the transmitted message itself without taking into consideration an adversarial effort to statistically infer private information using the cipher-text. It has been demonstrated that statistical inference can still be performed on encrypted data~\citep[For example][]{white2011phonotactic, miller2014know}. We complement this by capturing the notion of statistical inference of the private information $c$ from the transmitted data, and aim to find a way to minimize the ability of an adversary to infer $c$ using the transmitted data.

\subsection{Problem Definition}

Towards a more formal representation of the problem, we consider scenarios where
\begin{inparaenum}[\itshape i\upshape)]
\item \label{test} Bob's identity, $s$, is attached to any message that is sent by him;
\item there is no guarantee that the sent information will not be intercepted by an untrusted but passive eavesdropper Eve;
\item the information $x$ can be used to infer some private information $c$ about Bob; and
\item Alice knows the private information $c$ about Bob but Eve does not.
\end{inparaenum}
Under these assumptions, Bob would like to exploit the fact that Alice knows $c$ but Eve does not in order to send a message $z$ that is more useful to Alice than Eve. The utility value of the message follows the following decoding and ``hiding class" (HC) premises:
\begin{description}
\descitem{DECODING} Alice can make full use of the sent information $z$, i.e. obtain the original message $x$ from the transmitted message $z$; and
\descitem{HC} Eve's ability to make inference about $c$ given $s$, based on the sent information $z$ is minimized.
\end{description}

Formally, we use $\mathcal{S}$ for the set of identifiers of information providers, $\mathcal{I}$ for the information space and $\Sigma$ for the set of private classes (the private information about the information providers). Similarly, we define the random variables $S$ for the identifier of the information provider, $X$ for the piece of information that the provider would like to disclose, $C$ for the class that the provider belongs to and $Z$ for the encoded message that will be sent (called \emph{privatized information}), which is a function of the original information and the class. We call this function a \emph{privacy mapping function} and define it as $R : \Sigma \rightarrow \mathcal{I}^{\underline{\mathcal{I}}}$ where $\mathcal{I}^{\underline{\mathcal{I}}}$ is the set of injective functions $\mathcal{I} \rightarrow \mathcal{I}$. A simple way to think about $R$ is as an encoding scheme. That is, for every class $c \in \Sigma$, it outputs an encoding function for the input information $x$. Given $c \in \Sigma$, since $\m{R}{c}$ is injective, then there exists a left inverse $\m{R^l}{c}$ which will be used to decode the messages $z$ sent from subjects in class $c$.\footnote{We say that $g : D_2 \rightarrow D_1$ is a left inverse of a function $f : D_1 \rightarrow D_2$ if for all $x \in D_1$ we have $\m{g}{\m{f}{x}} = x$.} From that, $Z$ is simply equal to $\left[ R ( C ) \right](X)$. The statistical model that relates these random variables is described in~\cref{fig:graphical-model}.

For conciseness, in this paper we treat the case of continuous information spaces. Note that in the case of a discrete information space, the reader is instructed to follow the discussion by substituting probability density functions with probability mass functions for the distributions of $X$ and $Z$. Note that our treatment also covers the case of information spaces of mixed nature (that are discrete in some attributes and continuous in others) by using the appropriate probability distribution functions.

For the model in~\cref{fig:graphical-model}, one needs to supply the following probability distributions. $p(s)$, the prior of subjects transmitting messages in the system. $p(c|s)$, the adversary's prior of class membership for the different subjects (based on auxiliary knowledge). $p(x|c,s)$, the generative model of data given a class and a subject. Finally, $p(z|x,c)$ is simple and can be modeled as $\pr{Z=z|X=x,C=c} = 1$ if and only if $z = \left[R(c)\right](x)$ and $0$ otherwise, for all $z,x \in \mathcal{I}$ and $c \in \Sigma$.

Recall that the identity $s$ of the information provider is attached with the transmitted message. Moreover, the intended recipient knows the class $c$ of the information provider. Therefore, because of the injectivity requirement of the privacy mapping function, the intended recipient can decode the sent information $z$ back to the original message $x$. Hence the requirement \descref{DECODING} is satisfied.

Finally, in order to satisfy the second requirement \descref{HC} we would like to find a privacy mapping function $R$ that minimizes the amount of information that the privatized information $Z$ carries for the sake of inferring the private class $C$, given the subject identifier $S$, to an adversary. We adopt the measure of \emph{(conditional) mutual information} to model this quantity. We present the definition of conditional mutual information for continuous random variables, and refer the reader to~\citep[Definitions 2.61 and 8.54]{cover2006elements} for the corresponding definitions concerning discrete random variables and random variables that can be mixtures of discrete and continuous, respectively.

\newcommand{\mutualinfdef}{\citep[c.f. Definition 8.49]{cover2006elements}}
\begin{definition2}[\mutualinfdef] \label{eq:mi_def}
Let $X,Y$ and $Z$ be random variables\comment{ with a joint probability density function $f_{X,Y}(x,y)$ and marginal probability density functions $f_X(x)$ and $f_Y(y)$, respectively}. The \emph{conditional mutual information of $X$ and $Y$ given $Z$}, $I(X,Y|Z)$, is defined as
\comment{I(X,Y|Z) \triangleq E_{p_{X,Y,Z}(x,y,z)}\left[\log \frac{p_{X,Y|Z}(x,y|z)}{p_{X|Z}(x|z) p_{Y|Z}(y|z)}\right]}
$$
I(X,Y|Z) \triangleq E_{p(x,y,z)}\left[\log \frac{p(x,y|z)}{p(x|z) p(y|z)}\right]
$$
\end{definition2}

Intuitively, $I(Z,C|S;R)$ measures in bits, the expected amount of mutual information that the random variables $Z=\left[R(C)\right](X)$ and $C$ have, given the information in $S$.\footnote{The units are bits assuming the $\log$ base in~\cref{eq:mi_def} is $2$.} Mutual information also provides a sufficient and necessary condition for conditional independence as follows.
\def\currlemmaref{\citep[c.f. Corollary 2.92; c.f. Theorem 8.6.1]{cover2006elements}}
\begin{lemma2}[\currlemmaref] \label{lemma:mi_and_indep}
$I(Z,C|S;R) \geq 0$ for any privacy mapping function $R$. Furthermore, $I(Z,C|S;R) = 0$ if and only if $Z$ and $C$ are conditionally independent given $S$ using the privacy mapping function $R$.
\end{lemma2}

From the intuition above, and the fact in~\cref{lemma:mi_and_indep}, we set our objective to find a privacy mapping function $R$ that minimizes the conditional mutual information of the privatized information $Z$ and the private class $C$ given the identity of the information provider $S$ such that the model in~\cref{fig:graphical-model} holds. In short,
\begin{alignat*}{3}
R^* = &\argmin_{R}
& & I(Z,C \vert S; R) \numberthis \label{eq:main} \\
& \text{subject to }
& & R\text{ is a privacy mapping function} \\
& \text{and } & & \text{Model in~\cref{fig:graphical-model}}\\
\end{alignat*}

\vspace{-.25in}
Once a privacy mapping function $R$ is chosen, the communication process can be carried as follows.
\begin{description}
\item[Sending] The transaction of disclosing a piece of information $x \in \mathcal{I}$ by an information provider belonging to class $c \in \Sigma$ is performed by applying the following transformation $z \leftarrow \left[R(c)\right](x)$ and sending $z$ (or some encrypted version of it).
\item[Receiving] The transaction of receiving a piece of information $z \in \mathcal{I}$ sent by an information provider belonging to class $c \in \Sigma$ is performed by applying $x \leftarrow \left[R^l(c)\right](z)$. Where $R^l(c)$ is a left inverse of $R(c)$.
\end{description}

\comment{Because of the injectivity requirement for the privacy mapping function and the assumption that the intended recipient knows the class $c \in \Sigma$ to which the sender belongs, the process defined above allows the intended recipient to decode the transmitted message successfully, satisfying our first requirement in the problem definition.\todo[maybe give the requirements/assumptions labels?]}

Note that the problem in~\cref{eq:main} is not a convex problem. Furthermore, it is of interest to study how to learn the model in~\cref{fig:graphical-model} and find an optimal privacy mapping function $R$ from data. We will address this question in~\cref{sec:implementation}, but first we further study the properties of the formulated framework in the following section.

\comment{\begin{minipage}{\linewidth}
\begin{lstlisting}[label=lst:opt_complete,caption=Finding $\mathcal{R}(\cdot;\Theta)$ based on the model in~\cref{fig:graphical-model} by satisfying Condition~(),mathescape,frame=B]
Input: $\mathcal{I}$: Information space
Input: $\mathcal{S}$: Information providers set
Input: $\Sigma \triangleq \Sigma_{\text{intendend}}$: Provider class space
Input: $\mathcal{I_D} \subset \mathcal{I}^{\underline{\mathcal{I}}}$: Parametric search subspace
Input: $p(C|S), p(X|C,S), p(S)$: Model of the adversary
Output: $R : \Sigma \rightarrow \mathcal{I_D}$
minimize $\mathbb{E}_{p(Z,S)}\left[D_{KL}\left( p(C|Z,S) || p(C|S) \right)\right](\Theta)$
w.r.t	$\Theta$
s.t.	$R(\cdot;\Theta) \in \left(\Sigma \rightarrow \mathcal{I_D}\right)$
	$\forall c \in \Sigma, x\in \mathcal{I},z \in \mathcal{I}: p(Z=z|X=x,C=c) = \delta( \left[R(c;\Theta)\right](x) - z)$
	$p(C|Z,S) = \frac{\int_x \left[p(Z=z|X=x,C=c)\cdot p(X=x|C=c,S=s)\right]dx \cdot p(C=c|S=s)}{\sum_{\bar{c}} \int_x \left[p(Z=z|X=x,C=\bar{c})\cdot p(X=x|C=\bar{c},S=s)\right]dx \cdot p(C=\bar{c}|S=s)}$
	$p(Z=z,S=s) = ()$
\end{lstlisting}
\comment{s.t. 		$\forall c \in \Sigma : R(c;\Theta) \in \mathcal{I_D}$}
\end{minipage}}

\section{Further Analysis}
\label{sec:further_analysis}

First, we relate the value of the objective function in~\cref{eq:main} to Bayesian inference in the following lemma.

\begin{lemma}
If a privacy mapping function $R$ yields $I(Z,C|S;R) = 0$ then Bayesian inference of $C$ based on $Z$ is prevented for the adversary.
\end{lemma}
\begin{proof}
From~\cref{lemma:mi_and_indep} we know that $Z$ is conditionally independent of $C$ given $S$ which means $p(c|z,s) = p(c | s)$ which is the prior of the class membership that the adversary already possesses. Therefore, the disclosure of $Z$ does not change the adversary's belief regarding the private information $C$ given the subject identifier $S$.
\end{proof}

The next question that we need to ask is whether a privacy mapping function $R$ satisfying $I(Z,C|S;R) = 0$ is ever attainable. There are three reasons for this question. First, if such a privacy mapping function $R$ exists, then it means that by knowing $S$ (which is always attached to the message), $Z$ provides no extra information to inferring $C$ to an adversary, which sounds surprising. Second, there is generally a trade-off between information utility and privacy where optimal privacy is usually only attained at the cost of no utility~\citep[][]{dwork2006differential}. In our case, the utility of the information $Z$ to the intended recipient is always fully preserved, unrelated of the choice of $R$, since $\m{R}{c}$ is injective for all $c \in \Sigma$. From this it follows that the scenario of perfect privacy seems to be unattainable.\footnote{We consider ``perfect privacy" to be that the adversary's belief about $C$ given $S$ doesn't change after observing $Z$.} Finally, if such a privacy mapping function $R$ exists, it would assure optimality of~\cref{eq:main}. Fortunately (and somewhat unintuitively), such a mapping function can be attained as shown in the following sequence of results.

\begin{lemma} \label{lemma:if_c_doesnt_matter_then}
If there exists a function $f(z,s)$ such that $p(z|c,s)=f(z,s)$ for all $c \in \Sigma, z\in \mathcal{I}$ and $s \in \mathcal{S}$ then $p(z|s)\equiv f(z,s)$
\end{lemma}
\begin{proof}
$p(z|s) = \sum_{c \in \Sigma} p(z,c|s) = \sum_{c \in \Sigma} p(z|s,c)\cdot p(c|s) = \sum_{c \in \Sigma} f(z,s)\cdot p(c|s) = f(z,s)\cdot \sum_{c \in \Sigma} p(c|s) = f(z,s)$
\end{proof}

Using~\cref{lemma:if_c_doesnt_matter_then}, we prove the following theorem, which is a sufficient condition for optimality of~\cref{eq:main}.
\begin{theorem}\label{thm:sufficient_cond_opt}
If there exists a function $f(z,s)$ such that $p(z|c,s)=f(z,s)$ for all $c \in \Sigma, z\in \mathcal{I}$ and $s \in \mathcal{S}$ then $D_{KL}\left( p(c|z,s) || p(c|s) \right)=0$ for all $z \in \mathcal{I}$ and $s \in \mathcal{S}$.\footnote{\citep[Definition 8.46]{cover2006elements}: The Kullback-Leibler divergence is defined as $D_{KL}(p||q) = E_p \left[ \log \frac{p}{q} \right]$.}
\end{theorem}
\begin{proof}
Since $p(z|c,s)=f(z,s)$ then using~\cref{lemma:if_c_doesnt_matter_then} we know that $p(z|s)\equiv f(z,s)$. Therefore, for any $z \in \mathcal{I}$ and $s \in \mathcal{S}$ such that $f(z,s)=p(z|c, s)=p(z|s)\neq 0$ we get $\frac{p(c|z,s)}{p(c|s)}=\frac{p(z|c,s)\cdot p(c|s)}{p(c|s) \cdot p(z|s)}=\frac{p(z|c,s)}{p(z|s)}=1$. This implies $D_{KL}\left( p(c|z,s) || p(c|s) \right)=0$.
\end{proof}

\begin{corollary}\label{cor:diff_leading_to_optimal}
If a privacy mapping function $R$ achieves $p(z|c,s)=f(z,s)$ for some function $f(z,s)$, for all $c \in \Sigma, z\in \mathcal{I}$ and $s \in \mathcal{S}$ then $R$ is the optimal solution to~\cref{eq:main}.
\end{corollary}
\begin{proof}
The result follows from~\cref{thm:sufficient_cond_opt} and the fact that $I(Z,C|S;R) = E_{p(z,s)}[ D_{KL}(p(c|z,s;R)||p(c|s;R))]$.
\end{proof}

\comment{\Cref{thm:sufficient_cond_opt} is a valuable tool for proving optimality of privacy mapping functions. }Note that~\cref{thm:sufficient_cond_opt} is independent of the model of $p(c|s)$ (and $p(s)$). This is a very important observation since it means that in cases where a privacy mapping function $R$ satisfies the condition of the theorem, modeling the adversary's prior knowledge about information providers' class memberships is not needed. Furthermore, such privacy mapping function achieves perfect privacy against any adversary, regardless of her auxiliary knowledge $p(c|s)$ (or $p(s)$). In the following theorems we provide examples of using~\cref{thm:sufficient_cond_opt} that also serve as cases where such privacy mapping functions are attainable.

\begin{theorem} \label{thm:normal_opt_sol}
If $X|C=c,S=s \sim N(\mu_c, \Sigma_c)$ (Normal distribution) for every $c \in \Sigma$ and $s \in \mathcal{S}$, then $\left[R(c)\right](x) = \Sigma_c^{-\frac{1}{2}} \cdot \left( x-\mu_c \right)$ is an optimal solution to~\cref{eq:main}.
\end{theorem}
\begin{proof}
It is easy to verify that $Z|C=c,S=s \sim N(\bar{0}, I)$ for every $s \in \mathcal{S}$ and $c \in \Sigma$, where $\bar{0}$ is the origin in the information space (vector of zeros) and $I$ is the identity matrix (of the appropriate dimensions). This means that $p(z|c,s) \equiv f(z,s)$ (not a function of $c$). By using~\cref{thm:sufficient_cond_opt}, we therefore know that $R$ is the optimal solution to~\cref{eq:main}.
\end{proof}

The proofs of the following theorems are similar to this of~\cref{thm:normal_opt_sol} and were thus omitted for conciseness.

\begin{theorem}
If $X|C=c,S=s \sim Exp(\lambda_c)$ (Exponential distribution) for every $c \in \Sigma$ and $s \in \mathcal{S}$, then $\left[R(c)\right](x) = \lambda_c x$ is an optimal solution to~\cref{eq:main}.
\end{theorem}
\comment{\begin{proof}
It is easy to verify that $Z|C=c,S=s \sim Exp(1)$ for every $s \in \mathcal{S}$ and $c \in \Sigma$. This means that $p(Z=z|C=c,S=s) \equiv f(z,s)$ (not a function of $c$). By using~\cref{thm:sufficient_cond_opt}, we therefore know that $R$ is the optimal solution to~\cref{eq:main}.
\end{proof}}

\begin{theorem}
If $X|C=c,S=s \sim Gamma(k, \theta_c)$ (Gamma distribution with shape and scale parameters) for every $c \in \Sigma$ and $s \in \mathcal{S}$, then $\left[R(c)\right](x) = \frac{x}{\theta_c}$ is an optimal solution to~\cref{eq:main}.
\end{theorem}
\comment{\begin{proof}
It is easy to verify that $Z|C=c,S=s \sim Gamma(k, 1)$ for every $s \in \mathcal{S}$ and $c \in \Sigma$. This means that $p(Z=z|C=c,S=s) \equiv f(z,s)$ (not a function of $c$). By using~\cref{thm:sufficient_cond_opt}, we therefore know that $R$ is the optimal solution to~\cref{eq:main}.
\end{proof}}

\comment{
Given two vectors $u,v \in \mathbb{R}^N$, we define $w = v \div u$ as the vector $w \in \mathbb{R}^N$ such that $\forall i \in \set{1,2,\dots , N}, w_i = \frac{v_i}{u_i}$. That is, $v \div u$ is the element-wise division of $u$ over $v$.
}
\begin{theorem}
If $X|C=c,S=s \sim U(a_c, b_c)$ (Continuous Uniform distribution) for every $c \in \Sigma$ and $s \in \mathcal{S}$, then $\left[R(c)\right](x) = \frac{ x-a_c }{ b_c-a_c }$ is an optimal solution to~\cref{eq:main}.
\end{theorem}
\comment{\begin{proof}
It is easy to verify that $Z|C=c,S=s \sim U(0,1)$ for every $s \in \mathcal{S}$ and $c \in \Sigma$. This means that $p(Z=z|C=c,S=s) \equiv f(z,s)$ (not a function of $c$). By using~\cref{thm:sufficient_cond_opt}, we therefore know that $R$ is the optimal solution to~\cref{eq:main}.
\end{proof}}
\section{Implementation}
\label{sec:implementation}

In this section, we briefly describe an implementation of the learning problem that is publicly available in the form of a MATLAB\footnote{\url{https://www.mathworks.com/products/matlab/}} toolbox~\citep{aranki2015pditoolbox}. In this implementation, we investigate the question of learning a privacy mapping function $R$ from a labeled data set $\mathcal{D} = \set{(x_i,c_i)_i}$. This implies a simplifying assumption of ignoring the modeling of the random variable $S$ corresponding to the identity of the information providers. This assumption has the following implications on the model in~\cref{fig:graphical-model}. First, it implies that the adversary views information providers as uniformly distributed, that is $p(s) = \frac{1}{|\mathcal{S}|}$ for all $s \in \mathcal{S}$. Second, the assumption implies that the subject-class membership belief function of the adversary is equal for all subjects, that is $p(c \vert s) = p(c)$ for all $s \in \mathcal{S}$ and $c \in \Sigma$. As discussed in~\cref{sec:further_analysis}, in the cases where perfect privacy is achievable, the solutions are independent of these models and therefore these implications are not limiting. Further study is necessary to assess the level of privacy-degradation incurred by this assumption in cases of imperfect privacy. Third, this assumption implies that the generative model of data per class is independent of the subjects, that is $p(x|c,s) = p(x|c)$ for all $x \in \mathcal{I}, c \in \Sigma$ and $s \in \mathcal{S}$. Finally, $I(Z,C | S ; R)$ simplifies to $I(Z,C;R)$.

In order to make the problem in~\cref{eq:main} computationally tractable, a parametrized space for the privacy mapping functions can be introduced, allowing for the optimization to be performed on the parameter space. For example, consider the following parameter space
\begin{align*}
\Theta(n,\Sigma) = \left\{(A_c, b_c)_{c \in \Sigma} \vert \right. & \forall c \in \Sigma: A_c \in \real^{n \times n},\\
& \left. b_n \in \real^n, det(A_c) \neq 0 \right\}
\end{align*}

Then a parametrized space for affine privacy mapping functions on the classes set $\Sigma$ and information space $\mathcal{I}$ of dimension $n$ can be defined as
\begin{align*}
I_D(n,\Sigma,\mathcal{I}) = \left\{ R( \right. & \cdot; \theta) \vert \theta \in \Theta(n,\Sigma), R(\cdot;\theta) \in \left( \Sigma \rightarrow \mathcal{I}^{\underline{\mathcal{I}}} \right), \\
&\left. \forall c\in \Sigma : \left[R(c;\theta)\right](x) = A_c \cdot (x - b_c) \right\}
\end{align*}

Provided a parameter search space $\Theta$, the optimization problem in~\cref{eq:main} can be re-written as
\begin{equation}
\theta^* = \argmin_{\theta \in \Theta} I(Z,C ; R(\cdot; \theta)) \label{eq:main_par}
\end{equation}


The straightforward way to modeling the required distributions $p(c)$ and $p(x|c)$, from data, is non-parametrically by using high-dimensional histograms. This approach, while simple to implement, suffers from the curse of dimensionality as its complexity grows exponentially with the dimension of the information space. Once the models for $p(c)$ and $p(x|c)$ are constructed, the model for $p(z|c)$ can be computed for any choice of $\theta \in \Theta$ allowing the computation of the objective function in~\cref{eq:main_par}. Since the problem is non-convex, in order to optimize the objective function, we employ the genetic algorithm with the fitness function equal to the objective function in~\cref{eq:main_par}. The chosen selection policy is fitness-proportional while the chosen transformations (evolution/genetic) operators are both mutations and crossovers~\citep{banzhaf1998genetic}.

\section{Experimentation}
\label{sec:experiment}

In this section we walk the reader through an example that aims to motivate and demonstrate PDI. In this example we use data that are published by the Center for Disease Control and Prevention (CDC) as part of the National Health and Nutrition Examination Survey of 2012.\footnote{\url{https://wwwn.cdc.gov/nchs/nhanes/search/nhanes11\_12.aspx}} Specifically, we use the Body Measures (BMX\_G) portion of the data.\footnote{\url{https://wwwn.cdc.gov/nchs/nhanes/2011-2012/BMX\_G.htm}}

\subsection{Setting} \label{sec:experiment:subsec:setting}
In our setting, we consider the disclosed information to be both Body Mass Index (BMI) and weight. Our information providers are assumed to be individuals of both genders that are $19$ years of age or less. We consider the private information to be the weight status category of the subject. The CDC considers the following four standard weight status categories for the aforementioned age group
\begin{inparaenum}[\itshape i\upshape)]
\item underweight;
\item healthy weight;
\item overweight; and
\item obese.
\end{inparaenum}
There are $3355$ data points in the data set with subjects of $19$ years of age or less.

According to the definitions of the CDC, the BMI category of a child or a teen is classified based on the individual's BMI percentile among the same age and gender group as described in~\cref{tbl:BMI-for-age-table}. Since the age of the information provider is not part of the information space, the inference of the weight status category of the information provider based on BMI and weight is not perfect. The data for the different classes are depicted in~\cref{fig:BMI-weight-raw}.


\begin{table}[t]
\caption{BMI-for-age weight status categories and The corresponding BMI percentiles.} \label{tbl:BMI-for-age-table}
\begin{center}
\begin{tabular}{ll}
{\bf Weight Category}  &{\bf BMI Percentile Range} \\
\hline
Underweight & ~~~~~~~~~$BMI<5\%$ \\
Healthy Weight & ~$5\% \leq BMI < 85\%$ \\
Overweight & $85\% \leq BMI < 95\%$ \\
Obese & $95\% \leq BMI$
\end{tabular}
\end{center}
\vspace{-.25in}
\end{table}

\subsubsection{Inference Based on Original Data}
\label{sec:experiment:subsec:setting:sss:inference}

Using the data, we trained $3$ SVM classifiers with Gaussian kernels. The classifiers are aggregate in terms of the ``positive" class in the following sense. The first classifier treats the ``positive" class as the Underweight category (and so the ``negative" class is the rest of the categories). The second classifier treats the ``positive" class as either the underweight or healthy weight categories. Finally, the third classifier treats the ``positive" class as any category except the obese category. We used a $40-60$ split for training-testing. In numbers, we used $1371$ data points for training and $1984$ data points for testing.

The training for all SVMs was done using $10$-fold cross-validation among the data in the training set to pick the best $\sigma$ of the Gaussian kernels and the best box boundaries of the classifiers. The classification phase is done by taking a majority vote from the $3$ classifiers and the output is the class which most classifiers agree on. The results of the classifier are described in~\cref{tbl:results-raw} in terms of the confusion matrix of the different categories. The total accuracy of the classifier is $88.31\%$.\footnote{The adopted total accuracy measure is $trace(M)/N$ where $M$ is the confusion matrix and $N$ is the cardinality of the test set. This is the percentage of true classifications over the test set.}

\subsection{Privatizing Information} \label{sec:experiment:subsec:privatizing}

We would like to privatize the information at hand (BMI and weight) in order to maintain the weight status category as private as possible (based on the training set only). This scenario simulates a tele-monitoring scenario and fits the assumptions and motivation introduced in~\cref{sec:problem_formulation}. Therefore, we aim to utilize PDI in order to privatize the data as discussed earlier. In order to learn the privacy mapping function from the training data, we use the MATLAB toolbox mentioned in~\cref{sec:implementation}~\citep{aranki2015pditoolbox}. We used the affine privacy mapping functions for the parameterized search space as shown in the example in~\cref{sec:implementation}. Note that there are extra degrees of freedom in the problem, since any privacy mapping functions $R_1$ and $R_2$ related by $\forall c \in \Sigma: R_2(c) = A \cdot (R_1(c) - b)$ yield the same objective value in~\cref{eq:main} for any $A \in \real^{n \times n}, det(A) \neq 0$ and $b \in \real^n$. That is, applying the same injective affine transformation to all encoding functions in $R$ does not change the value of $I(Z,C|S;R)$. Therefore, in our problem we fix the encoding function of the ``underweight" class to the identity function, i.e. $\left[R(\text{``underweight"})\right](x) = x$.

\begin{table}[t]
\caption{Confusion matrix before privatizing. UW = Underweight, HW = Healthy Weight, OW = Overweight, OB = Obese} \label{tbl:results-raw}
\begin{center}
\begin{tabular}{cccccc}
& & \multicolumn{4}{c}{\bf Ground Truth Category} \\
\cline{3-6}
& &{\bf UW}  &{\bf HW} &{\bf OW} &{\bf OB} \\
\cline{3-6}
\multirow{4}{*}{\rotatebox[origin=c]{90}{\bf Predicted}~\rotatebox[origin=c]{90}{\bf Category}}
															&\multicolumn{1}{|c|}{\bf UW}
																& $47$ & $20$ & $0$ & $0$ \\
															&\multicolumn{1}{|c|}{\bf HW} 
																& $14$ & $1203$ & $66$ & $1$ \\
															&\multicolumn{1}{|c|}{\bf OW} 
																& $0$ & $45$ & $194$ & $47$ \\
															&\multicolumn{1}{|c|}{\bf OB} 
																& $0$ & $2$ & $37$ & $308$ 
\end{tabular}
\end{center}
\end{table}
\begin{figure}[t]
\centerline{
\includegraphics[width=.25\textwidth]{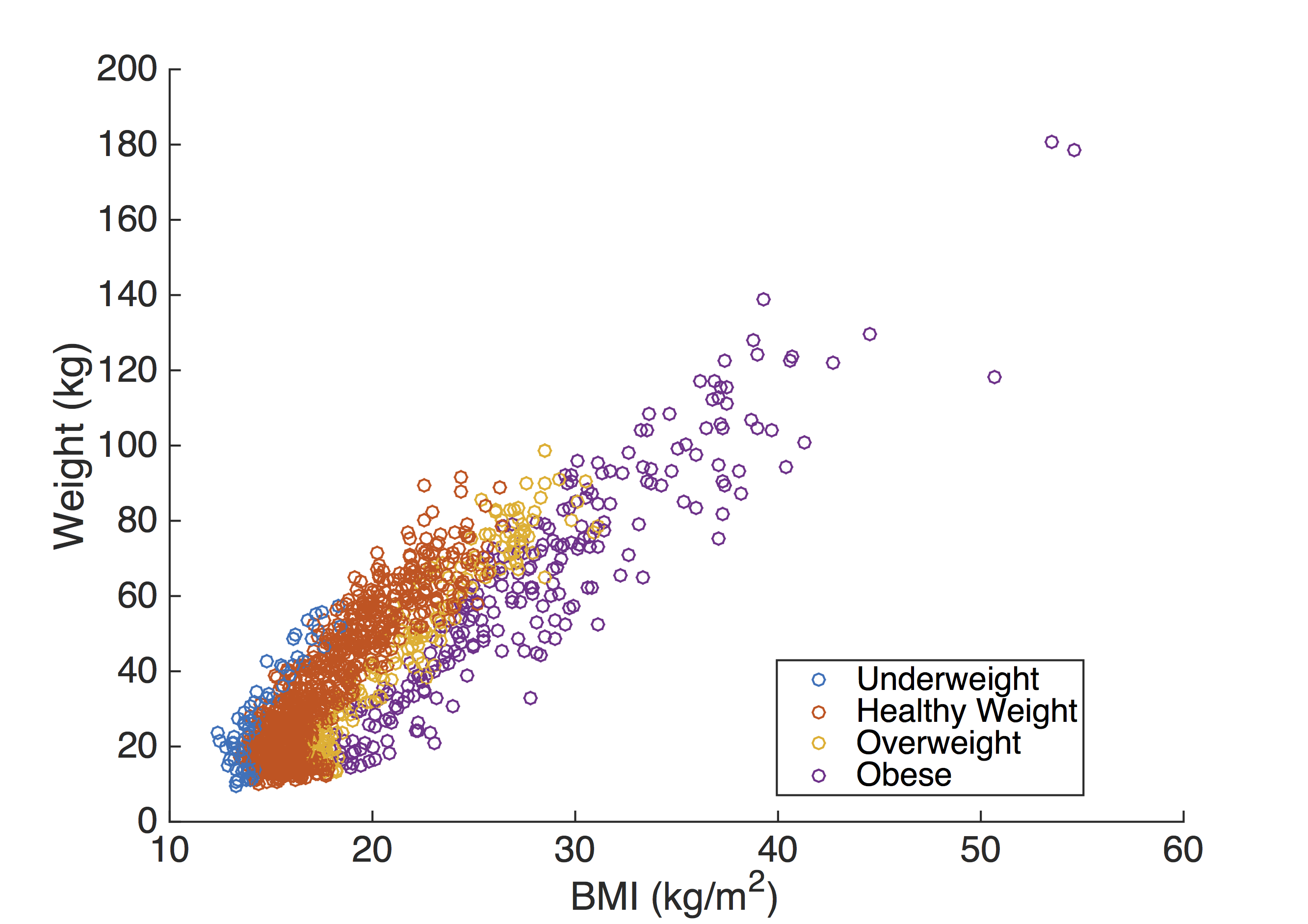}
}
\caption{BMI and weight for the different weight status groups.} \label{fig:BMI-weight-raw}
\end{figure}

The resultant privatized information is depicted in~\cref{fig:BMI-weight-private}. It is clear that it should be much harder to do inference of the weight category based on this privatized data, given the decreased distinguishability between classes. Note that calculating the privatized information is simple and efficient since now we know the parameters for the privacy mapping functions for the different classes.

\subsubsection{Inference Based on Privatized Data}

In order to evaluate the quality of the privatization, we now train new $3$ SVM classifiers with the same training procedure as in~\cref{sec:experiment:subsec:setting:sss:inference}, but this time using the privatized data (and of course, encoding the test set too for evaluation). Same as before, we then use a majority vote from the $3$ classifiers to predict the class of any data point. The resultant confusion matrix is described in~\cref{tbl:results-private}.

\begin{table}[t]
\caption{Confusion matrix after privatizing. UW = Underweight, HW = Healthy Weight, OW = Overweight, OB = Obese} \label{tbl:results-private}
\begin{center}
\begin{tabular}{cccccc}
& & \multicolumn{4}{c}{\bf Ground Truth Category} \\
\cline{3-6}
& &{\bf UW}  &{\bf HW} &{\bf OW} &{\bf OB} \\
\cline{3-6} 
\multirow{4}{*}{\rotatebox[origin=c]{90}{\bf Predicted}~\rotatebox[origin=c]{90}{\bf Category}}
															&\multicolumn{1}{|c|}{\bf UW}
																& $48$ & $14$ & $0$ & $5$ \\
															&\multicolumn{1}{|c|}{\bf HW} 
																& $13$ & $1217$ & $276$ & $290$ \\
															&\multicolumn{1}{|c|}{\bf OW} 
																& $0$ & $25$ & $13$ & $29$ \\
															&\multicolumn{1}{|c|}{\bf OB} 
																& $0$ & $14$ & $0$ & $32$ 
\end{tabular}
\end{center}
\end{table}
\begin{figure}[t]
\centerline{
\subfigure{\includegraphics[width=0.25\textwidth]{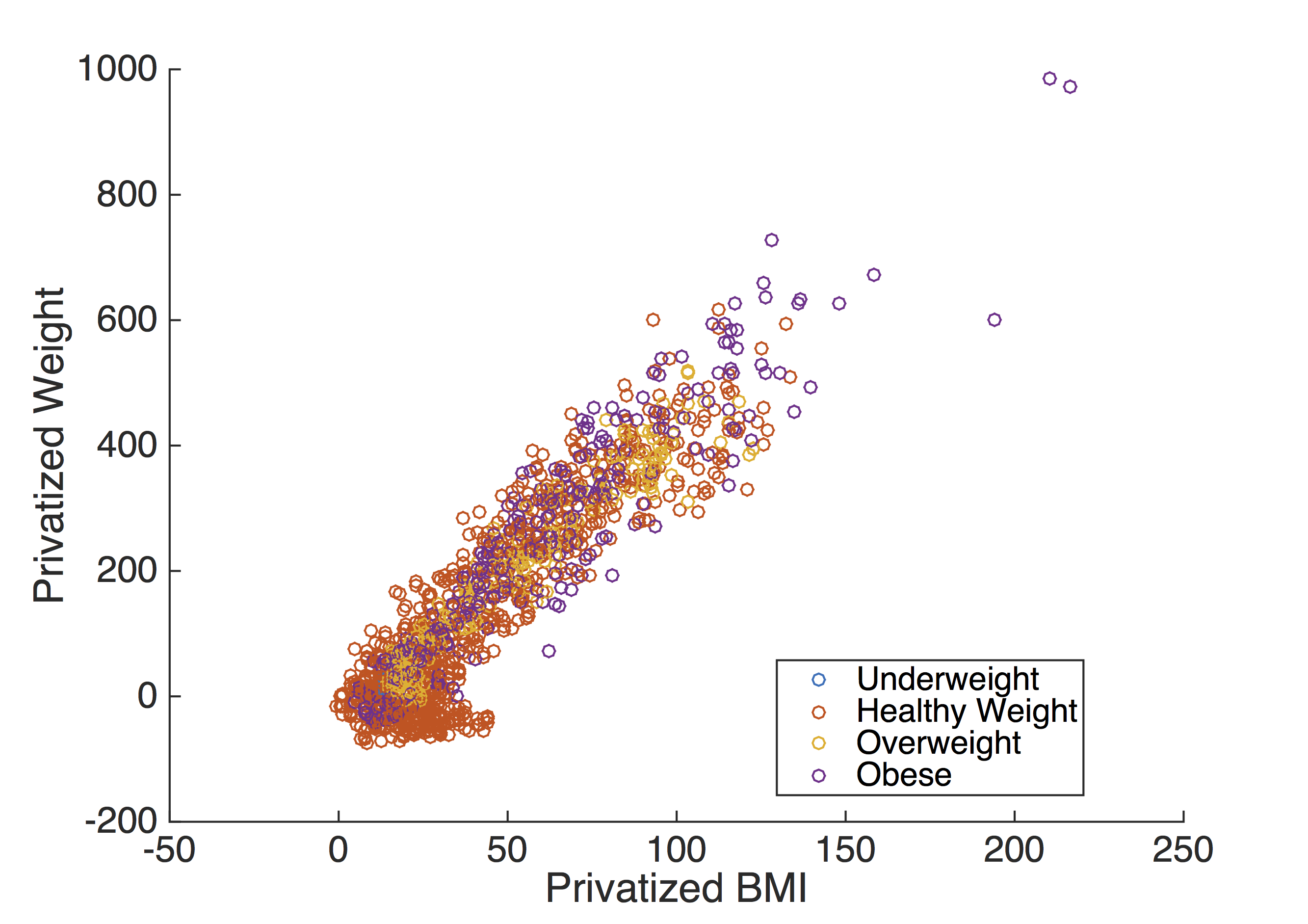}}
\subfigure{\includegraphics[width=0.25\textwidth]{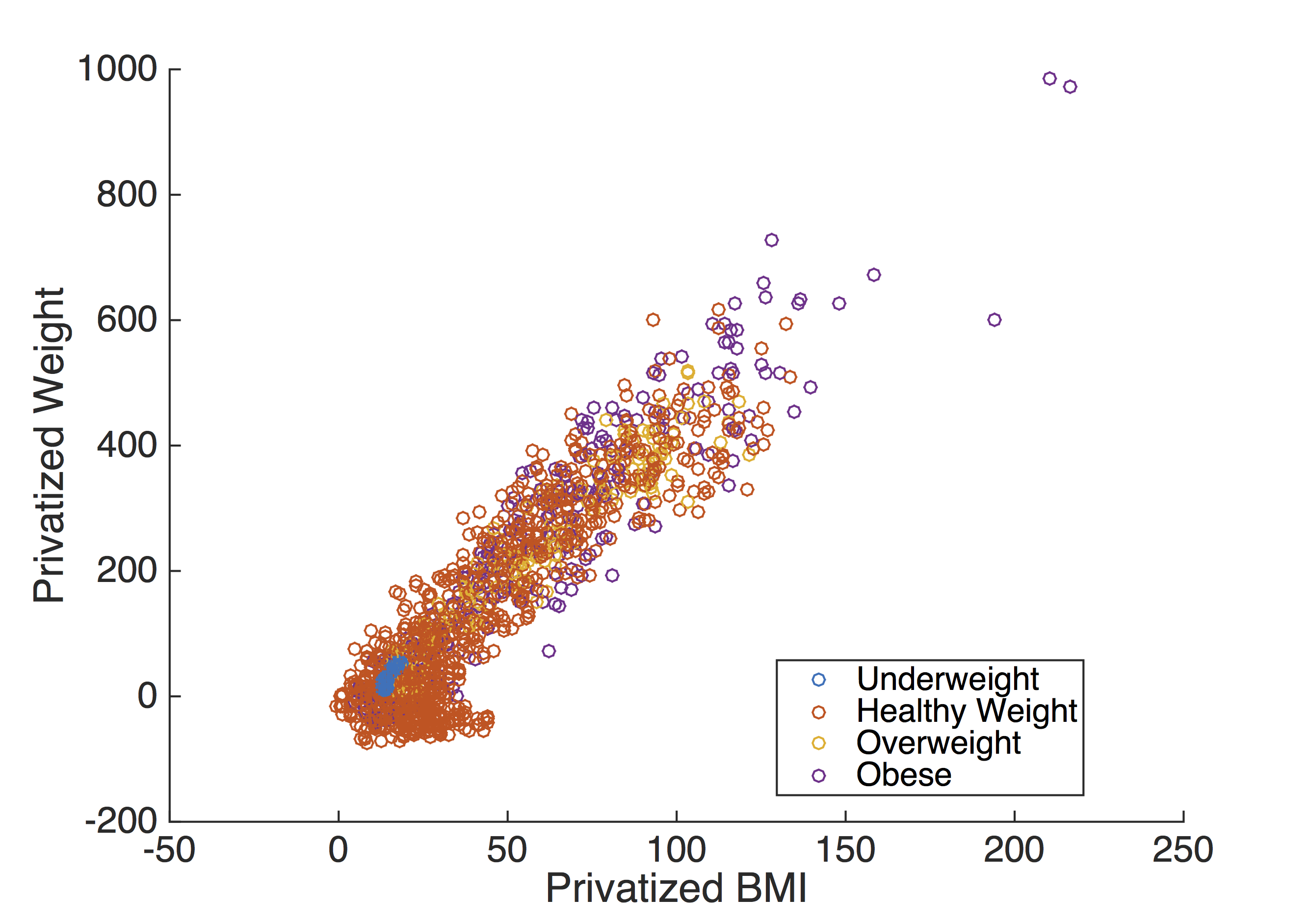}}
}
\caption{BMI and weight for the different weight status groups after privatization. The difference between the two plots is the order of plotting the different classes (for visual clarity).} \label{fig:BMI-weight-private}
\end{figure}

It is clear that the classification results are degraded after privatizing the information. The total accuracy dropped to $66.03\%$ (from $88.31\%$). Given that the data from different classes are highly indistinguishable, the classifier now classifies most data points as ``healthy weight". This is to be expected since most of the data points are in the ``healthy weight" category. In informal words, if a classifier would have to make a ``bet", it would bet on the class with the most amount of data points. Formally, a lower bound on the total accuracy can be achieved by considering the trivial classifier that always predicts ``healthy weight" (deterministic), which has total accuracy of $1270/1948 = 64.01\%$. This shows that our result of $66.03\%$ is not much further from a lower-bound guaranteed accuracy.

Note that the data set is biased in size against the ``underweight" category. There are only $126$ data points with weight category ``underweight" out of the $3355$ total data points ($3.76\%$). This makes privatizing that class particularly hard, especially because the modeling is based on $n$-dimensional histograms and is not parametric. For this reason the classification results before and after privatization for the ``underweight" category are comparable.

To intuitively demonstrate how privacy is preserved, we take a piece of privatized information at random from our data set, $z=[77.17, 296.45]^T$, without looking at its ground truth weight category. If we decode this data point using the decoding function of ``healthy weight", we get $x=[21, 53.8]^T$, which is a legitimate ``healthy weight" BMI and weight data point. If we use the decoding function of ``overweight", we get $x=[25.12, 62.4]^T$, which is also a legitimate ``overweight" BMI and weight data point. Similarly, if we use the decoding function of ``obese", we get $x=[30.42, 69.08]^T$, which is also a legitimate ``obese" BMI and weight data point.

\section{Discussion and Future Work}
\label{sec:discussion}

In this paper, we presented a view on privacy in which the data themselves \emph{need not} be the private object, but rather can \emph{be used} to infer private information. From this point of view, we derived a framework that preserves the privacy of the private information from being inferred from the communicated messages. We provided theoretical analysis and properties of the devised framework. An important result (\cref{thm:sufficient_cond_opt}) provided conditions that ensure perfect privacy while preserving full data utility. We showed that such conditions are achievable by providing closed-form solutions to some cases of data generative models.~\Cref{thm:sufficient_cond_opt} further showed that perfect privacy is not a function of the modeling of the adversary's auxiliary knowledge about the private information per subject, $p(c|s)$ (or $p(s)$). This observation is important because modeling adversary's auxiliary knowledge is generally a hard problem, and because it showed that perfect privacy can be achieved regardless of the adversary's auxiliary knowledge. That is, the same privatization protects information providers from all adversaries, regardless of their auxiliary knowledge.

Subsequently, we discussed an implementation of the learning problem resulting from the framework and demonstrated its use with a data set published by the Center for Disease Control and Prevention using data about individuals' Body Mass Indices, weights and their weight status categories. The experimentation shows that after privatizing the data set, the classification accuracy drops significantly, near a lower bound of guaranteed classification accuracy, thus achieving our set goal.

We make two important remarks about the approach presented in this paper. First, the described approach is philosophically different from the classical cryptography as it provides a model where the objective is maintaining the secrecy of the private information that is not the data themselves but the information that can be inferred based on the data. Second, even though the proposed approach is privacy-centric, it is not meant to serve as an alternative to cryptography but as a complement to it. That said, any message can be ``privatized" then encrypted. If the encryption is in that case compromised by an adversary getting access to the clear text message, the privacy is still preserved.

The current implementation of the devised learning problem suffers from the curse of dimensionality. The cost of learning grows exponentially with the number of dimensions of the information space. This is a result of our choice to model $p(z|c)$ as a multi-dimensional histogram. To make this framework practical, there is a need to study other ways of estimating the mutual information measure between the disclosed information and the private class. One appealing option is leveraging parametric learning and modeling each distribution $p(z|c)$ as a mixture model which could result in more computationally efficient estimation of the mutual information measure.

The presented framework has the potential of being extended to scenarios where the data recipient is not completely certain about the private class but is still more certain than the adversary. Such scenarios are clearly more general and may result in wider applicability of the framework to other scenarios than presented here. Indeed, in such scenarios, communicated messages can only be interpreted in a statistical sense and the implications of such assumptions must be studied as well.

Furthermore, the current implementation of the learning problem assumes that adversaries have equal belief about all the information providers so that the adversary's belief about $C$ is independent of $S$ and that the generative model of data $X$ per private class $C$ is independent of $S$. This is a simplifying assumption and its implications need to be further studied and remedied.

Given the non-convexity and the complexity of the problem at hand, areas for future research include studying heuristic techniques to learn the privacy mapping functions from sufficient and/or necessary conditions for local improvements in the mutual information as a function of local changes in the privacy mapping functions. This approach, as opposed to finding global optimal solutions to~\cref{eq:main}, is analogous to finding minimal anonymization as opposed to optimal anonymization in privacy preserving data publishing~\citep{fung2010privacy}.

\section*{Acknowledgments} 

We would like to thank Katherine Driggs Campbell for the initial conversation that spurred this idea. We are also greatly indebted to Gregorij Kurillo, Yusuf Erol and Arash Nourian for their fruitful discussions and feedback that significantly improved the quality of this paper. This work was supported in part by TRUST, Team for Research in Ubiquitous Secure Technology, which receives funding support for the National Science Foundation (NSF award number CCF-0424422).

\bibliographystyle{plainnat}
\bibliography{paper}

\begin{thebibliography}{30}
\providecommand{\natexlab}[1]{#1}
\providecommand{\url}[1]{\texttt{#1}}
\expandafter\ifx\csname urlstyle\endcsname\relax
  \providecommand{\doi}[1]{doi: #1}\else
  \providecommand{\doi}{doi: \begingroup \urlstyle{rm}\Url}\fi

\bibitem[Adam and Worthmann(1989)]{adam1989security}
Nabil~R Adam and John~C Worthmann.
\newblock Security-control methods for statistical databases: A comparative
  study.
\newblock \emph{ACM Computing Surveys (CSUR)}, 21\penalty0 (4):\penalty0
  515--556, 1989.

\bibitem[Aranki and Bajcsy(2015)]{aranki2015pditoolbox}
Daniel Aranki and Ruzena Bajcsy.
\newblock Private disclosure of information matlab toolbox, 2015.
\newblock URL \url{https://www.eecs.berkeley.edu/~daranki/PDI/}.

\bibitem[Aranki et~al.(2014)Aranki, Kurillo, Yan, Liebovitz, and
  Bajcsy]{aranki2014continouos}
Daniel Aranki, Gregorij Kurillo, Posu Yan, David Liebovitz, and Ruzena Bajcsy.
\newblock Continuous, real-time, tele-monitoring of patients with chronic
  heart-failure - lessons learned from a pilot study.
\newblock ICST, 11 2014.
\newblock \doi{10.4108/icst.bodynets.2014.257036}.

\bibitem[Banzhaf et~al.(1998)Banzhaf, Nordin, Keller, and
  Francone]{banzhaf1998genetic}
Wolfgang Banzhaf, Peter Nordin, Robert~E Keller, and Frank~D Francone.
\newblock \emph{Genetic programming: An introduction}, volume~1.
\newblock Morgan Kaufmann Publishers, Inc., 1998.

\bibitem[Bishop et~al.(2005)Bishop, Holmes, and Kelley]{bishop2005national}
Lynne Bishop, Bradford~J Holmes, and Christopher~M Kelley.
\newblock National consumer health privacy survey 2005.
\newblock \emph{California HealthCare Foundation, Oakland, CA}, 2005.

\bibitem[Chaudhry et~al.(2010)Chaudhry, Mattera, Curtis, Spertus, Herrin, Lin,
  Phillips, Hodshon, Cooper, and Krumholz]{chaudhry2010telemonitoring}
Sarwat~I Chaudhry, Jennifer~A Mattera, Jeptha~P Curtis, John~A Spertus, Jeph
  Herrin, Zhenqiu Lin, Christopher~O Phillips, Beth~V Hodshon, Lawton~S Cooper,
  and Harlan~M Krumholz.
\newblock Telemonitoring in patients with heart failure.
\newblock \emph{New England Journal of Medicine}, 363\penalty0 (24):\penalty0
  2301--2309, 2010.

\bibitem[Clark et~al.(2007)Clark, Inglis, McAlister, Cleland, and
  Stewart]{clark2007telemonitoring}
Robyn~A Clark, Sally~C Inglis, Finlay~A McAlister, John~GF Cleland, and Simon
  Stewart.
\newblock Telemonitoring or structured telephone support programmes for
  patients with chronic heart failure: Systematic review and meta-analysis.
\newblock \emph{BMJ}, 334\penalty0 (7600):\penalty0 942, 2007.

\bibitem[Cormode(2011)]{cormode2011personal}
Graham Cormode.
\newblock Personal privacy vs population privacy: Learning to attack
  anonymization.
\newblock In \emph{Proceedings of the 17th ACM SIGKDD international conference
  on Knowledge discovery and data mining}, pages 1253--1261. ACM, 2011.

\bibitem[Cover and Thomas(2006)]{cover2006elements}
Thomas~M Cover and Joy~A Thomas.
\newblock \emph{Elements of information theory}.
\newblock John Wiley \& Sons, 2 edition, 2006.

\bibitem[Denning(1976)]{denning1976lattice}
Dorothy~E. Denning.
\newblock A lattice model of secure information flow.
\newblock \emph{Commun. ACM}, 19\penalty0 (5):\penalty0 236--243, May 1976.
\newblock ISSN 0001-0782.
\newblock \doi{10.1145/360051.360056}.
\newblock URL \url{http://doi.acm.org/10.1145/360051.360056}.

\bibitem[Denning and Schlorer(1983)]{denning1983inference}
Dorothy~E. Denning and Jan Schlorer.
\newblock Inference controls for statistical databases.
\newblock \emph{Computer}, 16\penalty0 (7):\penalty0 69--82, 1983.

\bibitem[Duncan and Lambert(1989)]{duncan1989risk}
George Duncan and Diane Lambert.
\newblock The risk of disclosure for microdata.
\newblock \emph{Journal of Business \& Economic Statistics}, 7\penalty0
  (2):\penalty0 207--217, 1989.
\newblock \doi{10.1080/07350015.1989.10509729}.
\newblock URL
  \url{http://www.tandfonline.com/doi/abs/10.1080/07350015.1989.10509729}.

\bibitem[Duncan and Lambert(1986)]{duncan1986disclosure}
George~T Duncan and Diane Lambert.
\newblock Disclosure-limited data dissemination.
\newblock \emph{Journal of the American statistical association}, 81\penalty0
  (393):\penalty0 10--18, 1986.

\bibitem[Dwork(2006)]{dwork2006differential}
Cynthia Dwork.
\newblock Differential privacy.
\newblock In \emph{Automata, languages and programming}, pages 1--12. Springer,
  2006.

\bibitem[Dwork(2008)]{dwork2008differential}
Cynthia Dwork.
\newblock Differential privacy: A survey of results.
\newblock In \emph{Theory and Applications of Models of Computation}, pages
  1--19. Springer, 2008.

\bibitem[Farkas and Jajodia(2002)]{farkas2002inference}
Csilla Farkas and Sushil Jajodia.
\newblock The inference problem: A survey.
\newblock \emph{SIGKDD Explor. Newsl.}, 4\penalty0 (2):\penalty0 6--11,
  December 2002.
\newblock ISSN 1931-0145.
\newblock \doi{10.1145/772862.772864}.
\newblock URL \url{http://doi.acm.org/10.1145/772862.772864}.

\bibitem[Fung et~al.(2010)Fung, Wang, Chen, and Yu]{fung2010privacy}
Benjamin Fung, Ke~Wang, Rui Chen, and Philip~S Yu.
\newblock Privacy-preserving data publishing: A survey of recent developments.
\newblock \emph{ACM Computing Surveys (CSUR)}, 42\penalty0 (4):\penalty0 14,
  2010.

\bibitem[Giamouzis et~al.(2012)Giamouzis, Mastrogiannis, Koutrakis, Karayannis,
  Parisis, Rountas, Adreanides, Dafoulas, Stafylas, Skoularigis,
  et~al.]{giamouzis2012telemonitoring}
Gregory Giamouzis, Dimos Mastrogiannis, Konstantinos Koutrakis, George
  Karayannis, Charalambos Parisis, Chris Rountas, Elias Adreanides, George~E
  Dafoulas, Panagiotis~C Stafylas, John Skoularigis, et~al.
\newblock Telemonitoring in chronic heart failure: A systematic review.
\newblock \emph{Cardiology Research and Practice}, 2012, 2012.

\bibitem[Gkoulalas-Divanis et~al.(2014)Gkoulalas-Divanis, Loukides, and
  Sun]{gkoulalas2014publishing}
Aris Gkoulalas-Divanis, Grigorios Loukides, and Jimeng Sun.
\newblock Publishing data from electronic health records while preserving
  privacy: A survey of algorithms.
\newblock \emph{Journal of biomedical informatics}, 50:\penalty0 4--19, 2014.

\bibitem[Hsiao and Hing(2012)]{hsiao2012use}
Chun-Ju Hsiao and Esther Hing.
\newblock \emph{Use and characteristics of electronic health record systems
  among office-based physician practices, United States, 2001-2012}.
\newblock US Department of Health and Human Services, Centers for Disease
  Control and Prevention, National Center for Health Statistics, 2012.

\bibitem[Inglis(2010)]{inglis2010structured}
Sally Inglis.
\newblock Structured telephone support or telemonitoring programmes for
  patients with chronic heart failure.
\newblock \emph{Journal of Evidence-Based Medicine}, 3\penalty0 (4):\penalty0
  228--228, 2010.

\bibitem[Li et~al.(2007)Li, Li, and Venkatasubramanian]{li2007tcloseness}
Ninghui Li, Tiancheng Li, and Suresh Venkatasubramanian.
\newblock t-closeness: Privacy beyond k-anonymity and l-diversity.
\newblock In \emph{IEEE International Conference on Data Engineering},
  volume~7, pages 106--115, 2007.

\bibitem[Machanavajjhala et~al.(2007)Machanavajjhala, Kifer, Gehrke, and
  Venkitasubramaniam]{machanavajjhala2007ldiversity}
Ashwin Machanavajjhala, Daniel Kifer, Johannes Gehrke, and Muthuramakrishnan
  Venkitasubramaniam.
\newblock L-diversity: Privacy beyond k-anonymity.
\newblock \emph{ACM Trans. Knowl. Discov. Data}, 1\penalty0 (1), March 2007.
\newblock ISSN 1556-4681.
\newblock \doi{10.1145/1217299.1217302}.
\newblock URL \url{http://doi.acm.org/10.1145/1217299.1217302}.

\bibitem[Miller et~al.(2014)Miller, Huang, Joseph, and Tygar]{miller2014know}
Brad Miller, Ling Huang, Anthony~D Joseph, and J~Doug Tygar.
\newblock I know why you went to the clinic: Risks and realization of https
  traffic analysis.
\newblock \emph{arXiv preprint arXiv:1403.0297}, 2014.

\bibitem[{National Association of Health Data
  Organization}(1996)]{nahdo1996guide}
{National Association of Health Data Organization}.
\newblock A guide to state-level ambulatory care data collection activities,
  October 1996.

\bibitem[Sandhu(1993)]{sandhu1993lattice}
Ravi~S Sandhu.
\newblock Lattice-based access control models.
\newblock \emph{Computer}, 26\penalty0 (11):\penalty0 9--19, Nov 1993.
\newblock ISSN 0018-9162.
\newblock \doi{10.1109/2.241422}.

\bibitem[{State of California Office of Statewide Health Planning and
  Development}(2014)]{california2014california}
{State of California Office of Statewide Health Planning and Development}.
\newblock \emph{California inpatient data reporting manual, medical information
  reporting for California}, 7th edition, September 2014.

\bibitem[Sweeney(2002)]{sweeney2002k}
Latanya Sweeney.
\newblock k-anonymity: A model for protecting privacy.
\newblock \emph{International Journal of Uncertainty, Fuzziness and
  Knowledge-Based Systems}, 10\penalty0 (05):\penalty0 557--570, 2002.

\bibitem[Warner(1965)]{warner1965randomized}
Stanley~L Warner.
\newblock Randomized response: A survey technique for eliminating evasive
  answer bias.
\newblock \emph{Journal of the American Statistical Association}, 60\penalty0
  (309):\penalty0 63--69, 1965.

\bibitem[White et~al.(2011)White, Matthews, Snow, and
  Monrose]{white2011phonotactic}
Andrew~M White, Austin~R Matthews, Kevin~Z Snow, and Fabian Monrose.
\newblock Phonotactic reconstruction of encrypted voip conversations: Hookt on
  fon-iks.
\newblock In \emph{Security and Privacy (SP), 2011 IEEE Symposium on}, pages
  3--18. IEEE, 2011.

\end{thebibliography}
\vfill
\comment{
\vfill\pagebreak

\appendix
\section{Appendix -- Implementation and Toolbox Details}
\subsection{Introduction}
In this appendix, we describe in more details our implementation of the learning procedure for~\cref{eq:main_par} in MATLAB as a toolbox. The toolbox design was inspired by the structure and the convenience of coding in CVX~\citep{cvx,gb08}.\footnote{A MATLAB convex modeling and optimization framework, \url{http://cvxr.com/cvx/}} Therefore, the toolbox adopts a similar structure of coding and provides several ``keywords" for the users to use in order to define a PDI problem in a similar manner that one would define a CVX problem. We will present these keywords and their behavior in the following subsections. If the reader is familiar with CVX, then the similarities with CVX should be helpful to understand our toolbox.

The rest of this section is organized as follows. First we describe our data structures in~\cref{sec:impl:subsec:data-structures}. Then in~
\cref{sec:impl:subsec:using} we provide a short usage manual for the toolbox and finally in~\cref{sec:impl:subsec:the-engine} we describe the implementation of the learning engine for the privacy mapping functions.

\subsection{Data structures} \label{sec:impl:subsec:data-structures}

\subsubsection*{The PDI problem}
The data structure \myverb|pdi_problem| encapsulates the definition of the PDI problem at hand. This data structure holds all the necessary information about the parametrized search space, the constraints (if any) over the parameters of the search space, the data that will be used for modeling, the class definitions and the definitions of the dimensions of the information space at hand.

Two keywords were implemented to define a PDI problem definition block, namely \myverb|pdi_begin| and \myverb|pdi_end|. Any code that is related to the definition of the PDI problem at hand, as will be described in the following subsections, should be inserted between these two keywords. It is worth noting that nested PDI blocks are not allowed.
\comment{
\begin{lstlisting}[label=listing:skeleton,caption=Skeleton for a PDI problem definition in MATLAB]
PDI_begin 	% declare the beginning of
			% a PDI problem definition
			
	% PDI problem definition code goes
	% here
	
PDI_end 	% declare the end of a
			% PDI problem definition
\end{lstlisting}}

\subsubsection*{PDI variables}
A PDI variable is the atom object that the user can use to describe the parameters of the search space. For example, consider the search space $\mathcal{I_D} \triangleq \set{f \vert f(x;A,b) = A \cdot x - b, A \in \real^{N \times N}, b \in \real^{N}}$, the space of affine functions $\real^N \rightarrow \real^N$. Then, if the user has $k$ classes so that $\Sigma = \set{c_1, \dots , c_k}$, in order to represent a privacy mapping function $R(\cdot; \theta)$, parametrized by $\theta$, one would need a set of parameters $\theta_i \triangleq \left(A_i, b_i\right)$ for every $i \in \set{1,\dots k}$. Using $\theta \triangleq (\theta_i)_{i=1}^k$ one can describe the complete privacy mapping function. This results in $\left[R(c_i; \theta)\right](x) \triangleq f(x;\theta_i) \in \mathcal{I_D}$ for every $\sigma_i \in \Sigma$. PDI variables are implemented in the class \myverb|pdi_variable| to represent objects like $A_i$ and $b_i$.

\comment{For instance, think about $A_i$ and $b_i$ being PDI variables for each class $\sigma_i$. }The toolbox provides the keyword \myverb|pdi_var| that allows users to define variables for parameterization of the differential mapping functions. The syntax for using this keyword is as follows: \myverb|pdi_var varName(n, m)| which declares a PDI variable of size $n \times m$ with name \myverb|varName|. For convenience \myverb|pdi_var varName| is shorthand for \myverb|pdi_var varName(1,1)|. Currently, a PDI variable can be either a vector or a matrix.

\subsubsection*{PDI expressions and constraints}
In order to allow the representation of constraints over PDI variables (both convex and non-convex constraints), we implemented a data structure for PDI expressions. A PDI expression holds information about a mathematical expression that involves PDI variables, for example the expression \myverb|var1(1:2,[2 3])^2 - 3| where \myverb|var1| is a PDI variable (say of size $2 \times 3$ or larger) is a PDI expression. PDI expressions are implemented in the class \myverb|PDI_expression|.

PDI expressions have two main functions. The first function is that objects of the type \myverb|PDI_expression| can hold (potentially long) mathematical expressions involving PDI variables so that they can be used repeatedly. The second, and most important function is that PDI expressions are the building blocks of defining constraints over PDI variables. For example, if one of our PDI variables is a matrix \myverb|A| that we want to have determinant equal to one, by writing the line of code \myverb|det(A) == 1| inside the PDI problem block, the PDI engine will create a PDI expression for \myverb|det(A) - 1| and add a constraint on that expression (to be equal to $0$) to the \myverb|PDI_problem| object representing the PDI problem. These constraints are later passed to the learning engine so that it finds a feasible solution according to the user-defined constraints.

Since our problem is non-convex, there is no requirements for the constraints to be linear or even convex. However, since treating linear constraints is more efficient in general than treating non-linear constraints, linear constraints are tagged in the \myverb|PDI_problem| object so that they are passed separately to the learning engine for more efficient computation.

\comment{For visual convenience and for ease of reading the code, the keyword \myverb|subject_to| is provided to mark the beginning of the constraints block. The keyword is a void keyword that does nothing other than holding a line of code that makes the code look nicer.}

\subsection{Using the toolbox} \label{sec:impl:subsec:using}
\subsubsection*{More keywords}
The toolbox provides more keywords that can be used in the PDI problem definition block. The toolbox allows users to declare dimensions of the information space using the keyword \myverb|PDI_dimension|. As will be seen later, the current implementation models data and classes in a non-parametric way by binning the data into $n$-dimensional histograms, so the keyword also allows the user to define the bins to be used for modeling in that dimension. For example, \myverb|PDI_dimension weight 0:5:200| declares a dimension with the name ``weight" and with bins \myverb|0:5:200|.

Another building block of a PDI problem is a class (the private information). In order to declare a class, the keyword \myverb|PDI_class| can be used. For example, \myverb|PDI_class male| and \myverb|PDI_class female| declare the two classes ``male" and ``female", class names are one word strings (no spaces). For convenience, we allow declaration of multiple classes in one call of \myverb|PDI_class| by delimiting different class names by spaces. For instance the line \myverb|PDI_class male female| will declare both classes ``male" and ``female" in one line.

Sometimes, it is helpful to have a constant-like keyword that returns the total number of dimensions of the information space declared in the problem. This is for example useful when trying to define a variable that is of the same dimension of the information space. For that the keyword \myverb|PDI_nrdimensions| is provided. Similarly, it is useful to have a constant-like keyword that returns the total number of classes declared in the problem. For that the keyword \myverb|PDI_nrclasses| is provided. For example, the line \myverb|PDI_var b(PDI_nrdimensions, PDI_nrclasses)| will declare a PDI variable of size $N \times k$ assuming $N$ is the dimension of the declared information space and $k$ is the number of declared classes in the PDI problem.

In order to provide the data for the learning procedure (data per class), the keyword \myverb|PDI_datapoints| is provided. For example, \myverb|PDI_datapoints male male_data| provides the data stored in the variable \myverb|male_data| as coming from the class ``male" to the learning procedure. The convention we use is that a single data point is a column vector, so that in the example above \myverb|male_data| is expected to be of size $N \times m$ where $N$ is the dimension of the information space and $m$ is the number of data points provided.

A complete list of the keywords provided in the toolbox can be found in~\autoref{tbl:toolbox-keywords} and a complete example with results is provided in~\autoref{sec:experiment:subsec:privatizing}.

\begin{table*}[ht!]
\caption{List of Keywords And Their Descriptions} \label{tbl:toolbox-keywords}
\begin{center}
\begin{tabular}{| l | L{7.2cm} | l |}
\hline
{\bf Keyword and syntax} & {\bf Description} \\
\hline \hline
\myverb|PDI_start| & Begin a PDI problem definition block\\ \hline
\myverb|PDI_dimension <dimension name> <dimension bins>| & Declare a new dimension of information\\ \hline
\myverb|PDI_class <class1> <class2> ...| & Declare new classes of information providers\\\hline
\myverb|PDI_datapoints <class> <data expression>| & Provide data points from a class for the learning procedure\\\hline
\myverb|PDI_var <var 1>[(n1, m1)] <var 2>[(n2, m2)] ...| & Declare new PDI variables\\\hline
\myverb|PDI_reference <R(fv, cN)> <expression>| & Provide the (parametrized) differential mapping function\\\hline
\myverb|PDI_nrdimensions| & Returns the number of dimensions defined in the PDI problem\\\hline
\myverb|PDI_nrclasses| & Returns the number of classes defined in the PDI problem\\\hline
\myverb|PDI_end| & End a PDI problem definition block\\ \hline
\end{tabular}
\end{center}
\end{table*}

\subsubsection*{Structure of a PDI program}
We now describe the general rules that need to be satisfied in order to properly write a PDI program. A PDI program always starts with the keyword \myverb|PDI_begin|, followed by the PDI problem definition and always ends with the keyword \myverb|PDI_end|. In order to maintain consistency of data (in terms of the dimensions), it is assumed that all dimensions are declared before the first \myverb|PDI_datapoints| keyword is invoked. That said, once any data point is provided the information space dimensions are locked and cannot be edited any further. The reason is that the toolbox checks that the dimensions of the provided data points are consistent with the declared dimensions of the information space and therefore the toolbox assumes that all dimensions are declared beforehand. Failing to do so will result in an error thrown by the toolbox and the computation will be terminated.

Another rule concerning PDI variables and constraints is as follows. It is assumed that all PDI variables are declared before providing the parametrized reference function, i.e. before calling \myverb|PDI_reference| (to be introduced). Also, it is assumed that all PDI variables are declared before adding any constraints to the problem definition. The reason for the latter requirement is that a linear constraint is represented by the coefficients used to create it. For instance, if \myverb|x1|, \myverb|x2| and \myverb|x3| are PDI variables of size $1 \times 1$ each (for simplicity), the constraint \myverb|5 * x1 +  7 * x2 <= 11| is represented as the vector \myverb|[5 7 0 -11]|. In general, in order to map from a linear inequality (\myverb|<=|) constraint vector \myverb|v| to the symbolic representation, the following translation is used: \myverb|VARS * v(1:end-1) + v(end) <= 0|. From this, in order to keep the consistency of the constraints (in terms of dimensions), it is desired to know the number of variables before representing any constraints. Note that the last rule is an implication of the current implementation and can be later relaxed by fixing any existing constraints every time a new variable is declared. This can be done by appending zeros to the corresponding entries of the new variables in all existing linear constraint vectors. Although, we note that this rule isn't very limiting and therefore doesn't degrade the functionality of the toolbox.

The easiest way to remember these rules is by simply using the following order of things.
\begin{inparaenum}[1\upshape)]
\item Start a PDI problem definition block by stating \myverb|PDI_begin|;
\item declare all dimensions;
\item declare all classes and provide data for each class;
\item declare all PDI variables;
\item declare the reference function;
\item add all needed constraints; and
\item close the PDI problem definition block by using \myverb|PDI_end|.
\end{inparaenum}

\subsection{The engine} \label{sec:impl:subsec:the-engine}
The PDI engine is the entity that performs the learning of the parameters in the parameter space so that the mutual information between the differential information $Z$ and the class $C$ is minimized. In order to describe the engine, we will describe the way a user can declare the parametrized search space of differential mapping functions. For that, the keyword \myverb|PDI_reference| is used to declare the parametrized differential mapping function space. To best explain this keyword, we use the following example.

Consider the $\real^N \rightarrow \real^N$ affine functions space $\set{f \vert f(x;A,b) = A \cdot x - b, A \in \real^{N \times N}, b \in \real^{N}}$ and a situation where we have $k$ classes. First, we have to declare a matrix $A_i$ and a vector $b_i$ for each class. For the vectors $b_i$, one can stack them into a matrix \myverb|b| of size $N \times k$ such that the column \myverb|b(:,i)| is the $b_i$ vector corresponding to the class $i$. For the matrices $A_i$, we will represent them as a matrix \myverb|A| of size $N^2 \times k$ such that the column \myverb|A(:,i)| is the flattened matrix $A_i$ corresponding to class $i$, so that \myverb|reshape(A(:,i), N, N)| is our $A_i$. This can be done by declaring \myverb|PDI_var A(PDI_nrdimensions^2,PDI_nrclasses)| and \lstinline|PDI_var b(PDI_nrdimensions,PDI_nrclasses)|. Having these, we can write
\vspace{-.2in}
\begin{lstlisting}[breaklines]
PDI_reference @(xs, classN) bsxfun(@minus,
	reshape(
		A(:,classN), 
		PDI_nrdimensions, 
		PDI_nrdimensions
		) * xs),
	b)
\end{lstlisting}

In the code above, \myverb|@(xs, classN)| is used to define the function parameters, where the first one \myverb|xs| represents the data points passed to the function ($m$ data points of dimensions $N$ will be submitted as one call of with a $N \times m$ matrix) and \myverb|classN| is the class number (a single value). Each call to the reference function will include data only from one class.

Note that
\begin{inparaenum}[\itshape i\upshape)]
\item the reference function is assumed to be vectorized with respect to the first parameter, i.e. a matrix of data points will be passed to it in each call; and
\item the reference function body can include a call to an external function, i.e., users can design their own reference functions as regular MATLAB functions and use them in the function expression of \myverb|PDI_reference|.
\end{inparaenum}

The first step in the learning procedure is to prepare a prior for $C$ using the data provided by the user. This is modeled directly from the data as a histogram of data-class memberships. The next step in the engine is to represent an optimization objective function given the definition of a parameterized differential mapping function. Recall that the \myverb|PDI_reference| definition uses objects of type \myverb|PDI_variable| and note that it is provided to the toolbox as a string. For these two reasons, a compilation step is performed in order to translate the string representation of the reference function into MATLAB code and to resolve any PDI variables used in the function expression to their corresponding entries in the optimization vector of parameters that will be later passed to it by the optimizer. Once this step is done, we get a reference function that can be invoked with concrete values in place of the PDI variables used. This function is then used in order to calculate $\pr{Z|C}$, by translating $X$ to $Z$ using the function we compiled and then modeling the resulting distribution as an $N$-dimensional histogram. From $\pr{Z|C}$ and the $\pr{C}$ the objective function $I(Z,C;R)$ can be calculated and its value is returned as the objective value.

Using the objective function described above, the engine runs optimizers from the optimization toolboxes of MATLAB (by default, the engine first runs a genetic algorithm, \myverb|ga()|, then uses its output to initialize a gradient based optimizer, \myverb|fmincon()|) in order to find a set of parameters (\myverb|PDI_variable|'s) that yield an optimal (minimum) mutual information measure between the differential information $Z$ and the class $C$. That is, the engine solves the optimization problem in~\autoref{eq:main} with respect to the set of PDI variables declared and under the constraints provided on them in the PDI problem definition by the user. \comment{This is done by first compiling the provided reference function from the user (that uses PDI variables) into a usable function that assumes concrete values for these variables and then using this altered version in order to provide the optimizer with an objective function that maps the original information to the differential information using the compiled reference function (with the current parameters provided by the optimizer) and outputs the mutual information $\m{I}{C,Z}$.} In the end, the engine substitutes the PDI variables declared by the user by the values found by the optimizer in the user's workspace (so that they become number matrices instead of objects of type \myverb|PDI_variable|).

Note again that in the current implementation, the distribution $\pr{Z|C}$ is modeled non-parametrically as a high-dimensional histogram using computed bins based on the original bins provided by the user when invoking \myverb|PDI_dimension| and the reference function provided by the user. This approach clearly suffers from the curse of dimensionality but serves as a simple first implementation for a proof of concept. We discuss approaches to amend this problem in~\autoref{sec:future}.
}
\comment{
\subsection{Information Theory Background} \label{appendix:information_theory}
The Shannon entropy~\citep{shannon1948mathematical} and differential entropy.
\begin{definition2}[\citep{shannon1948mathematical}]
The \emph{Shannon entropy} of a discrete random variable $X$, denoted by $H(X)$, is defined by $H(X) = \expec{-\log p(x)}$.
\end{definition2}

\newcommand{\diffentrpy}{\citep[Definition 8.1]{cover2006elements}}
\begin{definition2}[\diffentrpy]
The \emph{Differential entropy} of a continuous random variable $X$ with density $f(x)$, denoted by $h(X)$, is defined by $h(X) = \expec{-\log f(x)}$.
\end{definition2}

\newcommand{\kldist}{\citep[c.f. Definition 8.46]{cover2006elements}}
\begin{definition2}[\kldist]
Let $P$ and $Q$ be two probability measures such that $P$ is absolutely continuous with respect to $Q$. The \emph{relative entropy} (or \emph{Kullback-Leibler distance}) $D_{KL}(P||Q)$ between two probability measures $P$ and $Q$ is defined by $D_{KL}(P || Q) = \int \log \frac{dP}{dQ} dP = E_P[\log \frac{dP}{dQ}]$
\end{definition2}

\subsection{Omitted Proofs}
\begin{proof}[Proof of~\cref{cor:diff_leading_to_optimal}]
It is sufficient to show that
$$I(Z,C|S;R) = \expec{D_{KL}\left( p(C|Z=z,S=s) || p(C|S=s) \right)}$$

We abuse notation and denote by $E_x[f(x)]$ the expectation of $f(x)$ w.r.t. the distribution $P_X(x)$ and by $E_{x|Y=y}[f(x,y)]$ the expectation of $f(x,y)$ w.r.t. the distribution $P_{X|Y=y}(x)$ . We write
\begin{align*}
& I(Z,C|S=s;R) = \\
= & E_s[D_{KL}(\pr{Z,C|S=s;R} || \pr{Z|S=s;R}\pr{C|S=s;R})] = \\
= & E_s[E_{z,c|S=s}[ - log \frac{\pr{Z=z,C=c|S=s;R}}{\pr{Z=z|S=s;R}\pr{C=c|S=s;R}}]] = \\
= & E_{z,c,s}[ - log \frac{\pr{Z=z,C=c|S=s;R}}{\pr{Z=z|S=s;R}\pr{C=c|S=s;R}}] = \\
= & E_{z,c,s}[ - log \frac{\pr{C=c|Z=z,S=s;R}}{\pr{C=c|S=s;R}}] = \\
= & E_{z,s}[E_{c|Z=z,S=s}[ - log \frac{\pr{C=c|Z=z,S=s;R}}{\pr{C=c|S=s;R}}]] = \\
= & E_{z,s}[ D_{KL}(\pr{C|Z=z,S=s;R}||\pr{C|S=s;R})]
\end{align*}
as requested.
\end{proof}
}
\comment{
\subsection{Notes}
\todo[Needs to be removed!]

\begin{enumerate}
\item Auxiliary data-sources are an issue in privacy. Example: Netflix prize~\citep{narayanan2006break}. \citeauthor*{adam1989security} call it ``supplementary knowledge"~\citeyearpar{adam1989security}.
\item Inference-sensitive scenarios is what we care about.
\item Collusion between different adversaries doesn't degrade the performance of the system.
\item Bias is one of the criteria that \citeauthor*{adam1989security} evaluated, the additive noise solutions to differential privacy seem to suffer from this problem~\citep{dwork2006differential,dwork2008differential}. For reference,~\citeauthor*{matloff1986another} shows the following result in non-interactive statistical databases (SDBs) with additive noise. If $X$ is the original value of a statistical query and $X` = X + \alpha$ is the perturbed value of the statistical query and if $X$ is a positive numerical variable with a strictly decreasing pdf then $\expec{X | X' = w} < w$~\citep{matloff1986another,adam1989security}. Although,~\citeauthor*{adam1989security} hint that the problem is less severe in output-perturbation methods (like interactive SDB, differential privacy)~\citeyearpar[Section 5]{adam1989security}.
\item \citeauthor*{matloff1986another} further shows the following bias result in non-interactive SDBs with additive noise. Let $X$ and $Y$ be correlated attributes with a bi-variate Gaussian distribution whose expected value is $0$. Furthermore, let $X' = X + \alpha$ be the perturbed value of attribute $X$ where $\alpha$ is an independent noise with mean $0$ and variance $Var(\alpha)$, then the following bias occurs. Let
\begin{align*}
\m{m}{w} &\triangleq \expec{Y|X=w} \\
&= \left( \expec{XY}/Var(X) \right) w \\
\m{m'}{w} &\triangleq \expec{Y|X'=w} \\
&= \left( \expec{X'Y}/Var(X') \right) w
\end{align*}

Since $\alpha$ is independent of $Y$ and both $Y$ and $\alpha$ have a mean of zero we get
\begin{align*}
\m{m'}{w} &= \frac{\expec{XY}}{Var(X')} w \\
&= \frac{\expec{XY}Var(X)}{Var(X')Var(X)} w \\
&= \frac{Var(X)}{Var(X')} \m{m}{w}
\end{align*}
or equivalently
\begin{align*}
\m{m'}{w} &= \frac{Var(X)}{Var(X + \alpha)} \m{m}{w} \\
&= \frac{Var(X)}{Var(X) + Var(\alpha)} \m{m}{w} \\
&= \frac{1}{1 + Var(\alpha)/Var(X)} \m{m}{w}
\end{align*}
which implies
\begin{equation}
\frac{\m{m'}{w}}{\m{m}{w}} = \frac{1}{1 + Var(\alpha)/Var(X)}
\end{equation}

Therefore, if $Var(\alpha) = Var(X)$, which is no unreasonable for perturbing noise, a bias of $50\%$ occurs~\citep{matloff1986another,adam1989security}. More conservatively, if $Var(\alpha) = \frac{1}{2} Var(X)$ then we get $\frac{\m{m'}{w}}{\m{m}{w}} = 66.667\%$ which is $33.333\%$ bias.
\item The small query-set problem (that is exact and partial compromise/disclosure problems)~\citep[Section 4.2.1]{adam1989security}. Does this also apply to differential privacy? According to \citeauthor*{adam1989security}, it might~\citeyearpar[Section 5.3.3, p. 543]{adam1989security}.
\item In their survey, \citeauthor*{adam1989security} consider the notion of ``confidential attribute(s)" but focus mostly on the problem of exact compromise/disclosure as opposed to partial compromise/disclosure as we are presenting in PDI~\citep{adam1989security}.
\item ``Computer security is not privacy protection"~\citep{sweeney2002k}
\end{enumerate}
}
\end{document}